%% file: arxiv.tex
\documentclass[11pt, a4paper]{article}
\usepackage[left=25mm,right=25mm,top=25mm,bottom=30mm]{geometry}
\usepackage{framed,latexsym,amssymb,amsmath,amsthm,float,cite,complexity,setspace}
\usepackage[shortlabels]{enumitem}
\usepackage[ddmmyyyy]{datetime}
\usepackage[shortend, vlined]{algorithm2e}

\newtheorem{theorem}{Theorem}[section]
\newtheorem{lemma}[theorem]{Lemma}
\newtheorem{corollary}[theorem]{Corollary}
\newtheorem{claim}{Claim}
\newtheorem{proposition}[theorem]{Proposition}

\newcommand{\independentcutset}{\textsc{Independent Cutset}}

\newcommand{\dist}{\operatorname{dist}}
\newcommand{\calT}{\mathcal{T}}
\newcommand{\calO}{\mathcal{O}}
\newcommand{\cut}{\operatorname{cut}}

\newcommand{\version}[2]{#1}
\newcommand{\QED}{} 

\begin{document}
\onehalfspace

\title{Exact and Parameterized Algorithms for the \\ Independent Cutset Problem}
\author{Johannes Rauch$^1$ \and Dieter Rautenbach$^1$ \and U\'{e}verton S. Souza$^2$}
\date{}

\maketitle
\vspace{-10mm}
\begin{center}
	{\small 
		$^1$ Institute of Optimization and Operations Research, Ulm University, Ulm, Germany\\
		$^2$ Instituto de Computa\c{c}\~{a}o, Universidade Federal Fluminense, Niter\'{o}i, Brazil\\
		\texttt{$\{$johannes.rauch,dieter.rautenbach$\}$@uni-ulm.de},
		\texttt{ueverton@ic.uff.br}
	}
\end{center}

\begin{abstract}
\input{sections/abstract}

\end{abstract}

\input{sections/introduction}
\input{sections/preliminaries}
\input{sections/exact-algorithms}
\input{sections/parameterized-algorithms}

\bibliographystyle{plain}
\bibliography{algorithms}

\newpage
\appendix
\input{sections/algorithm-domindcs2}
\input{sections/dynamic-programming}

\end{document}

%% file: sections/abstract.tex
The \independentcutset{} problem asks whether there is a set of vertices in a given graph that is both independent and a cutset.
Such a problem is \NP{}-complete even when the input graph is planar and has maximum degree five.
In this paper, we first present a $\calO^*(1.4423^{n})$-time algorithm to compute a minimum independent cutset (if any).
Since the property of having an independent cutset is MSO$_1$-expressible, our main results are concerned with structural parameterizations for the problem considering parameters incomparable with clique-width.
We present \FPT{}-time algorithms for the problem considering the following parameters:
the dual of the maximum degree, the dual of the solution size, the size of a dominating set (where a dominating set is given as an additional input), the size of an odd cycle transversal, the distance to chordal graphs, and the distance to $P_5$-free graphs. 
We close by introducing the notion of $\alpha$-domination, which allows us to identify more fixed-parameter tractable and polynomial-time solvable cases.

%% file: sections/introduction.tex
\section{Introduction} \label{chap-intro}
\version{Connectivity of graphs and the existence of vertex or edge cutsets with certain properties are fundamental topics within graph theory.}{}
The {\sc Matching Cutset} problem is a well-studied problem in the literature, both from a structural and from an algorithmic point of view.
It asks whether a graph $G$ admits a set of edges $M$ such that $G-M$ is disconnected and no two distinct edges of $M$ are incident.
A natural variation of this problem is obtained by replacing the word ``edges'' by ``vertices'' and the word ``incident''  by ``adjacent'' in the previous problem definition.
Doing this yields the \independentcutset{} problem, which is also known as the {\sc Stable Cutset} problem.
\version{A formal definition together with more preliminaries is given in Section~\ref{chap-pre}.}{}
Tucker~\cite{tucker1983coloring} studied \independentcutset{} in the context of perfect graphs and graph colorings in 1983.
In 1993, Corneil and Fonlupt~\cite{corneil1993stable} explicitly asked for the complexity of \independentcutset{}. 
They studied the problem in the context of perfect graphs, too.

It is not hard to see that a graph $G$ with minimum degree at least two has a matching cutset if and only if the line graph of $G$ has an independent cutset.
Therefore, the first \NP-completeness proof of \independentcutset{} is due to Chv\'{a}tal~\cite{chvatal1984recognizing}, who presented in 1984 the first \NP-completeness proof for the {\sc Matching Cutset} problem.
Brandstädt, Dragan, Le and Szymczak~\cite{BrDrLeSz} showed in 2000 that \independentcutset{} stays \NP-complete, even when restricted to $K_4$-free graphs.
Note that the problem is trivial on $K_3$-free graphs, since the neighborhood of any vertex constitutes an independent cutset.
They also concluded that it is \NP-complete on perfect graphs.
Le and Randerath~\cite{LeRa} proved in 2003 that \independentcutset{} is \NP-complete on 5-regular line graphs of bipartite graphs.
In 2008, Le, Mosca and Müller~\cite{LeMoMu} showed that the problem is \NP-complete on planar graphs with maximum degree five.

On the other hand, several polynomial time solvable cases have been identified.
In 2002, Chen and Yu~\cite{chen_yu_2002} answered a question by Caro in the affirmative by showing that all graphs with $n$ vertices and at most $2n-4$ edges admit an independent cutset.
Their proof can be used for a polynomial time algorithm that finds such a set.
Le and Pfender~\cite{LePf} characterized in 2013 the extremal graphs having $2n-3$ edges but no independent cutset.
In particular, they showed that \independentcutset{} can be decided for graphs with $n$ vertices and $2n-3$ edges in polynomial time. 
Le, Mosca and Müller~\cite{LeMoMu} also showed that the problem can be decided in polynomial time for claw-free graphs of maximum degree 4, $\{\text{claw}, K_4\}$-free graphs, and claw-free planar graphs.
The general case for maximum degree four is still open.
Some more polynomial time solvable cases can be found in~\cite{BrDrLeSz}.

Regarding the parameterized complexity of \independentcutset{},
Marx, O'Sullivan and Razgon~\cite{MaOsRa} in 2013 considered the \textsc{Independent $s$-$t$ Cutset} problem: Given a graph $G$ and two vertices $s$ and $t$ of $G$, the question is whether there exists an independent cutset that separates $s$ and $t$ in $G$.
They showed that the problem of finding a minimum independent $s$-$t$ cutset in a graph can be solved in $\calO(f(k) \cdot (n+m))$ time, where $n$ is the number of vertices of the input graph, $m$ is the number of edges of the input graph, and $k$ is an upper bound on the solution size. 
They transform the input graph to a graph of treewidth bounded by $g(k)$, where $g$ is some function only depending on $k$.
The transformed graph preserves all $s$-$t$ cutsets of size at most $k$ of the input graph.
An application of Courcelle's Theorem~\cite{Co} then shows the fixed-parameter tractability.
Beside that, to the best of our knowledge, there is no other work about parameterized algorithms for \independentcutset{} in the literature.

Moreover, although a $\calO^*(2^n)$ time algorithm for \independentcutset{} is trivial, there is also no discussion on more efficient exact exponential-time algorithms.
Motivated by this, we present an $\calO^*(3^{n/3})$ time algorithm for \independentcutset{} in Section~\ref{chap_exact}.
It is based on iterating through all maximal independent sets of a graph.
Note that $3^{1/3} < 1.4423$.
In the same section we show how to adapt this algorithm to compute a minimum independent cutset in the same time (if there is one).

From the parameterized complexity point of view, an algorithmic meta-theorem of Courcelle, Makowsky and Rotics~\cite{Co00} states that any problem expressible in monadic second-order logic (MSO$_1$) can be solved in $\calO(f(cw) \cdot n)$ time, where $cw$ is the clique-width of the input graph.
Originally this required a clique-width expression as part of the input. 
This restriction was removed when Oum and Seymour~\cite{OUM2006514} gave an \FPT{} algorithm, parameterized by the clique-width of the input graph, that finds a $2^{\mathcal{O}(cw)}$-approximation of an optimal clique-width expression. 
Since the property of having an independent cutset can be expressed in MSO$_1$, it holds that \independentcutset{} is in \FPT{} concerning clique-width parameterization.
Therefore, our focus in this work is on structural parameters that measure the distance from the input instance to some ``trivial class'' that is relevant to the problem and does not have bounded clique-width.
\version{We present several parameters for which \independentcutset{} is fixed-parameter tractable in Section~\ref{chap_param}.
First we consider two dual parameters in Section~\ref{sec_dual}, the dual of the maximum degree and the dual of the solution size.
Then we prove our most important result in Section~\ref{sec_dom}, where we consider a variant of \independentcutset{} with a dominating set as additional input, and the cardinality of the dominating set is the parameter.
Next we measure the distance of the input graph to three graph classes, and consider the distance as a parameter.
The graph classes in question are bipartite graphs in Section~\ref{sec_bipartite}, chordal graphs in Section~\ref{sec_chordal}, and $P_5$-free graphs in Section~\ref{sec_p5}.
Finally we generalize the results of Section~\ref{sec_p5} in Section~\ref{sec_gen_p5}.
With our results we were also able to identify more polynomial time solvable cases.}{}

\version{}{\input{sections/ostar}}

%% file: sections/ostar.tex
We use the $\calO^*$-notation to suppress polynomial factors in the $\calO$-notation.

%% file: sections/preliminaries.tex
\section{Preliminaries} \label{chap-pre}

We introduce important notions and notations in this section.
Unless stated otherwise, we consider only finite, simple, loopless and undirected graphs,
and use standard terminology.
Let $G$ be a graph and let $v$ be a vertex of $G$.
We denote the \emph{neighborhood} of $v$ in $G$ by $N_G(v)$.
The \emph{closed neighborhood} $N_G[v]$ of $v$ in $G$ is $N_G(v) \cup \{v\}$.
Let $S$ be a set of vertices of $G$.
Then $N_G[S]$ is the union of all closed neighborhoods of vertices in $S$, and $N_G(S) = N_G[S] \setminus S$.
We denote the graph induced by $S$ with $G[S]$.
We write $G-S$ for the graph $G[V(G) \setminus S]$.
The set $S$ is a \emph{cutset} of $G$ if $G-S$ is disconnected. It is an \emph{independent set} of $G$ if its vertices are pairwise nonadjacent in $G$.
An \emph{independent cutset} of $G$ is both a cutset and an independent set of $G$.
Let $A$ and $B$ be two sets of vertices of $G$, and assume that $S$ is a cutset of $G$.
We say that $S$ \emph{splits} $A$ in $G$ if the elements of $A \setminus S$ belong to at least two different components of $G - S$.
We say that $S$ \emph{separates} $A$ and $B$ in $G$ if $A \setminus S$ and $B \setminus S$ belong to different components of $G-S$.
A \emph{vertex cover of $G$} is a set of vertices $X$ such that for every edge $uv \in E(G)$ we have $u \in X$ or $v \in X$.
A \emph{dominating set of $G$} is a set of vertices $X$ such that every vertex of $G$ is in $X$ or has a neighbor in $X$.
Let $H$ be a graph and let $\mathcal{H}$ be a graph class.
We say that $G$ is \emph{$H$-free} if $G$ does not contain $H$ as an induced subgraph.
We say that $G$ is \emph{$\mathcal{H}$-free} if $G$ is $H$-free for every $H \in \mathcal{H}$.

Throughout this article, we consider the following decision problem, which we will parameterize appropriately for our results.
\begin{center} \fbox{
\begin{tabular}{l}
\independentcutset{} \\
Instance: A connected graph $G$. \\
Question: Does $G$ have an independent cutset?
\end{tabular}}
\end{center}
\noindent Although we formulated \independentcutset{} as a decision problem, our algorithms will also return an independent cutset of $G$ in the affirmative case, or are easily adaptable to do so.


A \emph{tree decomposition} of a graph $G$ is a pair $\calT = (T, \{X_t:t \in V(T)\})$ such that the following conditions hold. The graph $T$ is a tree and $X_t$, $t \in V(T)$, is a set of vertices of $G$. We refer to the vertices of $T$ exclusively as \emph{nodes}. The sets $X_t$ are called \emph{bags}. Furthermore,
\begin{enumerate}[(T1)]
\item $V(G) = \bigcup_{t \in V(T)} X_t$, \label{T1}
\item for every edge $uv \in E(G)$ there is a bag $X_t$ such that $u, v \in X_t$, and \label{T2}
\item for every vertex $v \in V(G)$, the set $\{t \in V(T): v \in X_t\}$ induces a (connected) subtree $T_v$ of $T$. \label{T3}
\end{enumerate}
The \emph{width} of a tree decomposition equals $\max_{t \in V(T)} |X_t| - 1$.
The \emph{treewidth} of a graph $G$ is the minimum possible width of a tree decomposition of $G$.

It is helpful to consider the trees of tree decompositions as rooted trees. 
This means we distinguish a special node $r$ of $T$. We call $r$ the \emph{root} of $T$. 
This introduces a natural parent-child and ancestor-descendant relation on the tree.
A node of $T$ that has degree 1 and is not the root is a \emph{leaf}.

A tree decomposition $\calT = (T, \{X_t:t \in V(T)\})$ is \emph{nice} if
\begin{enumerate}[(N1)]
\item $X_r = \emptyset$ and $X_\ell = \emptyset$ for every leaf $\ell$ of $T$, and
\item each node of $T$ is of one of the following types:
\begin{itemize}
\item (Introduce node) A node $t$ with exactly one child $t'$ such that $X_t = X_{t'} \cup \{v\}$ for some $v \notin X_{t'}$. We say that $v$ is introduced by $t$.
\item (Forget node) A node $t$ with exactly one child $t'$ such that $X_t = X_{t'} \setminus \{w\}$ for some $w \in X_{t'}$.
We say that $w$ is forgotten by $t$.
\item (Join node) A node $t$ with exactly two children $t'$ and $t''$ such that $X_t = X_{t'} = X_{t''}$.
\end{itemize}
\end{enumerate}
Given a tree decomposition of width $k$, a nice tree decomposition of width $k$ can be computed efficiently, as the following lemma shows.

\begin{lemma}[Lemma 7.4 in~\cite{parameterized_algorithms_book}] \label{lem_nice}
If a graph $G$ admits a tree decomposition of width at most $k$, then it also admits a nice tree decomposition of width at most $k$. Moreover, given a tree decomposition $\calT$ of $G$ of width at most $k$, one can in time $\calO(k^2 \cdot \max(n(G), n(T))$ compute a nice tree decomposition of $G$ of width at most $k$ that has at most $\calO(kn(G))$ nodes.
\end{lemma}

A \emph{clique tree} of a graph $G$ is a tree decomposition of $G$ where the bags are exactly the maximal cliques of $G$.
A graph is \emph{chordal} if every induced cycle has length exactly three. 
We will use the following fact about chordal graphs.

\begin{theorem}[Theorem 3.1 in~\cite{blair_peyton_1993}] \label{thm_cliquetree}
A connected graph $G$ is chordal if and only if there exists a clique tree of $G$. If a connected graph $G$ is chordal, a clique tree of $G$ can be determined in polynomial time.
\end{theorem}

\version{\input{sections/preliminaries-refined-td}

\input{sections/ostar}}{}

%% file: sections/preliminaries-refined-td.tex
Dallard, Milanič and Štorgel~\cite{dallard_milanic_storgel_2022} introduced a special kind of tree decomposition.
Given a nonnegative integer $\ell$, an \emph{$\ell$-refined tree decomposition} of a graph $G$ is a pair $\hat\calT = (T, \{(X_t, U_t):t \in V(T)\})$ such that $\calT = (T, \{X_t:t \in V(T)\})$ is a tree decomposition, and for every $t \in V(T)$ we have $U_t \subseteq X_t$ and $|U_t| \leq \ell$. 
We say $\calT$ is the \emph{underlying tree decomposition} of $\hat\calT$. 
We extend any concept defined for tree decompositions to $\ell$-refined tree decompositions by considering them on the underlying tree decomposition.
The residual independence number of $\hat\calT$ is defined as $\max_{t \in V(T)} \alpha(G[X_t \setminus U_t])$ and denoted by $\hat\alpha(\hat\calT)$.

%% file: sections/exact-algorithms.tex
\section{An Exact Exponential Algorithm} \label{chap_exact}
In this section, we present a single-exponential time algorithm for finding a minimum independent cutset (if any).
It is structured as follows.
First, we present a preliminary structural result in Lemma~\ref{lem_maximal}.
Corollary~\ref{cor_exact_general} is an immediate consequence of Lemma~\ref{lem_maximal}.
By combining these results with some known results from the literature, we are able to solve \independentcutset{} in $\calO^*(3^{n/3})$ time, which is given in Corollary~\ref{cor_exact}.
Furthermore, we get another polynomial-time solvable case, which is stated in Corollary~\ref{cor_2k2}.
We conclude this section with Lemma~\ref{lem_min_indcs} and Corollary~\ref{cor_min_indcs}, where we show how to adapt the algorithm to find a minimum independent cutset (if there is one).

We start with an easy but important structural observation.

\begin{lemma} \label{lem_maximal}
If a connected graph $G$ has an independent cutset $S$, then every independent set $S'\supseteq S$ is also a cutset of $G$.  
\end{lemma}
\begin{proof}
Let $G$ be a connected graph, and let $S$ be an independent cutset of $G$. Let $S'\supset S$ be another independent set of $G$.
Since $G$ is connected and $S$ is a cutset, any component of $G-S$ has a vertex with a neighbor in $S$. This implies that any component of $G-S$ has a vertex not in $S'$. Thus, $S'$ is also a cutset of $G$. 
\QED \end{proof}

Lemma~\ref{lem_maximal} implies that it suffices to consider maximal independent sets to decide \independentcutset{}. This is fact is used in Corollary~\ref{cor_exact_general}.

\begin{corollary} \label{cor_exact_general}
Let $G$ be a connected graph with $n$ vertices.
If there is an algorithm that enumerates all maximal independent sets of $G$ in $\calO^*(f(n))$ time, then \independentcutset{} with $G$ as input is solvable in $\calO^*(f(n))$ time.
\end{corollary}
\begin{proof}
This follows from Lemma~\ref{lem_maximal} and the fact that checking if a given set $S$ is a cutset of $G$ can be done in polynomial time.
\QED \end{proof}

By a result of Moon and Moser~\cite{moon_moser_1965}, a graph with $n$ vertices has $\calO(3^{n/3})$ maximal independent sets. 
Johnson, Yannakakis and Papadimitrou~\cite{johnson_et_al_1988} showed that all maximal independent sets can be enumerated with $\calO(n^3)$ delay. 
This, together with Corollary~\ref{cor_exact_general}, gives a fast exponential algorithm for the problem, and a graph class for which the problem is efficiently solvable.

\begin{corollary} \label{cor_exact}
\independentcutset{} can be solved in $\calO^*(3^{n/3})$ time.
\end{corollary}

\begin{corollary}\label{cor_2k2}
\independentcutset{} on $2K_2$-free graphs can be solved in polynomial time.    
\end{corollary}
\begin{proof}
$2K_2$-free graphs with $n$ vertices only have $\calO(n^2)$ maximal independent sets~\cite{farber_1989}.
\QED \end{proof}

We showed in the proof of Lemma~\ref{lem_maximal} that, in order to decide whether a graph admits an independent cutset, it suffices to consider all maximal independent sets.
In the following lemma we show how to obtain small independent cutsets from maximal ones. 

\begin{lemma}\label{lem_min_indcs}
Given a connected graph $G$ and an independent cutset $S'$ of $G$, one can compute in polynomial time the smallest independent cutset $S$ of $G$ contained in $S'$.
\end{lemma}

\begin{proof}
Let $G$ and $S'$ be as in the statement.
We construct a hypergraph $H$. The vertices of $H$ are the vertices of $S'$ and the components of $G-S'$. 
The hyperedges of $H$ are
\begin{align*}
S' \cup \{ K \,|\, \text{$K$ is a component of $G-S'$ adjacent to $v$} \} \text{ for every } v \in S'.
\end{align*}
It is easy to see that a minimum edge cut of $H$ corresponds to a smallest independent cutset $S$ being a subset of $S'$. 
Since a minimum edge cut can be computed in polynomial time on hypergraphs, see for example~\cite{queyranne_1995}, we can find such a set $S'$ in polynomial time.
\QED \end{proof}

The next corollary follows from Lemma~\ref{lem_maximal} and Lemma~\ref{lem_min_indcs}. 

\begin{corollary} \label{cor_min_indcs}
Given a connected graph $G$ that admits an independent cutset, a minimum independent cutset of $G$ can be found in $\calO^*(3^{n/3})$ time.
\end{corollary}

%% file: sections/parameterized-algorithms.tex
\section{Parameterized algorithms} \label{chap_param}

Intuitively, one could expect the following.
If a graph has few edges, then there always exists an independent cutset, and it is easy to find one.
This was made precise by Caro's Conjecture, and it was proved by Chen and Yu~\cite{chen_yu_2002}.
If a graph has many edges, then independent cutsets (if there are any) must have a specific structure, or there is no independent cutset at all.
This can be seen by the given parameterizations presented in this paper, under which \independentcutset{} is fixed-parameter tractable.
Indeed, all our parameterizations have in common that they are small if the input graph (or at least a part of it) is ``dense''.

The section is structured as follows.
In Subsection~\ref{sec_dual}, we consider the dual of the maximum degree and the dual of the solution size as a parameter.
What follows in Subsection~\ref{sec_dom} is our most important result, where we consider \independentcutset{} with a dominating set as an additional input, and the size of the dominating set is the parameter.
Then we consider the distance by vertex removals from the input graph to three different graph classes and take such distances as parameters in Subsections~\ref{sec_bipartite}, \ref{sec_chordal} and \ref{sec_p5}. 
The graph classes under consideration are bipartite graphs, chordal graphs, and $P_5$-free graphs.
Finally, we generalize the distance to $P_5$-free graphs results in Subsection~\ref{sec_gen_p5}.

\subsection{Dual parameterizations} \label{sec_dual}

We consider the dual of the maximum degree as a parameter in Theorem~\ref{thm_Deltaindcs}.
Actually, it also follows from Theorem~\ref{thm_domindcs} that the problem is fixed-parameter tractable with respect to this parameter, because the size of a minimum dominating set of a graph $G$ with $n$ vertices and maximum degree $\Delta$ is at most $n - \Delta$.
Nevertheless, we give a direct proof of this fact, since it comprises a faster algorithm.

\begin{theorem} \label{thm_Deltaindcs}
Let $G$ be the connected input graph with $n$ vertices and maximum degree $\Delta$, and let $\calO^*(f(n))$ be the running time of an algorithm enumerating all maximal independent sets of $G$.
It holds that \independentcutset{} can be solved in $\calO^*(2^k\cdot f(k))$ time, where $k = n - \Delta$.
\end{theorem}
\begin{proof}
Let $G$ be a connected graph with $n$ vertices, let $v$ be a vertex of $G$ having maximum degree, and let $R = V(G) \setminus N_G[v]$. 

Any independent cutset of $G$ containing $v$ is contained in $R\cup\{v\}$, since $|R|=k$, we can enumerate them in $\calO(2^k)$ time. 
So, we may assume from here that every independent cutset of $G$ (if any) does not contain $v$. 

Let $S^*$ be a minimal independent cutset of $G$ that does not contain $v$.

Observe that $N_G[v] \setminus S^*$ is nonempty and belongs to one component of $G-S^*$.
Thus, some subset $\emptyset \neq R' \subseteq R \setminus S^*$ belongs to the other components.
Such a set $R'$ can be guessed in $\calO(2^k)$ time.  
Assuming that we are dealing with the correctly guessed set $R'$, the cutset $S^*$ can be seen as a minimal independent cutset separating $v$ from $R'$.
Let $I = N_G(v) \cap N_G(R')$. The set $I$ must be a subset of $S^*$; in particular, the set $I$ must be independent in $G$.

Note that $S^*$ cannot contain a vertex $w\in N_G(v)\setminus I$; otherwise, as $S^*$ is a minimal cutset, there is a path $P$ from $v$ to $R'$ such that $V(P)\cap S^*=\{w\}$ and $V(P)\cap (R\setminus R')\neq \emptyset$. This implies that a vertex of $R\setminus R'$ is in neither $S^*$ nor in the same component as $v$ after the removal of $S^*$, contradicting the fact that $R'$ is the correctly guessed set.

Thus, after guessing $R'$, we can contract $I$ into a single vertex $x_I$ and remove $N_G[v]\setminus I$. At this point, the reduced graph has size $k$, and we can enumerate its maximal independent sets in $\calO(f(k))$ time, one of them must contain $(S^*\setminus I) \cup \{x_I\}$. Let $S'$ be such a set. By Lemma~\ref{lem_maximal}, the set $(S'\setminus \{x_I\})\cup I$ is also a independent cutset of $G$ (recall that $S'\supseteq S^*$). This concludes the proof. 
\QED \end{proof}

Recall that $\calO^*(2^k \cdot f(k))$ is faster than the running time $\calO^*(3^k)$ of Theorem~\ref{thm_domindcs}, since $f(k) = \calO(3^{k/3})$ by Corollary~\ref{cor_exact}.

\medskip

We finish this subsection by considering a lower bound on the dual of the solution size as a parameter.
\version{A \emph{vertex cover} of a graph $G$ is a set of vertices $X$ such that for every $uv \in E(G)$ we have $u \in X$ or $v \in X$.}{}

\begin{theorem}
Let $G$ be the connected input graph with $n$ vertices, and let $\calO^*(f(k))$ be the running time of an algorithm enumerating all vertex covers of size at most $k$ of $G$. 
Then there is an algorithm that correctly decides in $\calO^*(f(k))$ time if $G$ has an independent cutset of size at least $n - k$.
\end{theorem}
\begin{proof}
Let $G$, $n$ and $k$ be as in the statement.
Assume that there is an independent cutset $S^*$ of $G$ such that $|S^*| \geq n - k$.
The set $V(G) \setminus S^*$ is a vertex cover of size at most $k$ of $G$.
Therefore it suffices to iterate over all vertex covers $X$ of size at most $k$ of $G$, and check whether $V(G) \setminus X$ is an independent cutset of $G$.
By assumption, this can be done in $\calO^*(f(k))$ time.
\QED \end{proof}

We remark that all vertex covers of size at most $k$ can be enumerated in $\calO^*(2^k)$ time.

\subsection{Dominating set} \label{sec_dom}

In this subsection, we prove a central theorem of our paper.
\version{A \emph{dominating set} of a graph $G$ is a set of vertices $X$ such that every $v \in V(G)$ has a neighbor in $X$ or is itself in $X$.}{}
We consider \independentcutset{} with a dominating set $X$ of $G$ as an additional input.
We show that this variant is fixed-parameter tractable when parameterized by $|X|$.
We split the proof of this fact over three lemmata.
Among other things, we distinguish the cases when $X$ is split by an independent cutset or not.
In Lemma~\ref{lem_domindcs1} we settle the case when $X$ is not split by an independent cutset.

\begin{lemma} \label{lem_domindcs1}
Let $G$ be a connected graph, and let $X$ be a dominating set of $G$. 
Assume that there is an independent cutset $S^*$ of $G$ such that $X \setminus S^*$ is contained in at most one component of $G - S^*$.
Then there is an algorithm that returns an independent cutset of $G$ in $\calO^*(2^k)$ time, where $k = |X|$.
\end{lemma}
\begin{proof}
Let $G$, $X$, $S^*$ and $k$ be as in the statement.
In $\calO^*(2^k)$ time, we can check whether any subset of $X$ is an independent cutset of $G$.
Therefore, we may assume from here that $S^*$ is not a subset of $X$. Since $X$ dominates $G$, the set $X$ is not a proper subset of $S^*$. In particular, $S^* \neq X$.

To find an independent cutset of $G$, we iterate over all disjoint partitions $(A, X')$ of $X$ such that $A$ is nonempty and $X'$ is independent in $G$.
In one iteration, we guess $X \cap S^*$ as $X'$.
Let $F =N_G(X') \setminus A$ and $I = N_G(A) \setminus (X' \cup F)$.
Note that $(A, X', F, I)$ is a partition of $V(G)$, because $X$ is a dominating set of $G$.
A vertex of $I$ is either in $S^*$ or in the same component as a vertex of $A$.
Let $B$ be the set of vertices that are not in a component with a vertex of $A$ in $G - (X' \cup I)$.
Note that $B \subseteq F$, and since $S^*$ is an independent cutset of $G$ and $X\setminus S^*$ is contained in at most one component of $G - S^*$, it holds that $B \neq \emptyset$.
This implies that there is a component $K$ of $G[B]$ with $N_G(K) \cap I \subseteq S^*$.
For such a component $K$, the set $X' \cup (N_G(K) \cap I)$ is an independent cutset of $G$.
Since all this can be done in $\calO^*(2^k)$ time, this finishes the proof.
\QED \end{proof}

For the case when the dominating set is split by an independent cutset, we need to consider the following situation.
Let $G'$ be a connected graph whose vertex set is the disjoint union of four sets $A$, $B$, $N_A$ and $N_B$ such that
\begin{enumerate}[(i)]
	\item $A$ and $B$ are nonempty,
	\item $A \cup B$ is independent in $G'$, and
	\item $N_{G'}(A) = N_A$ and $N_{G'}(B) = N_B$. (Note that $N_{G'}(A) \cap N_{G'}(B) = \emptyset$.)
\end{enumerate}
We show how to decide efficiently if $G'$ has an independent cutset $S' \subseteq N_A \cup N_B$ that separates $A$ and $B$ in $G$ by testing the satisfiability of a {\sc 2-SAT} formula.
Let $H$ be the bipartite graph induced by the set of edges between $N_A$ and $N_B$, and let $K_1, \dots, K_r$ be the vertex sets of the components of $H$.
Let $N_{A,i} = K_i \cap N_A$ and $N_{B,i} = K_i \cap N_B$ for all $i\in [r]$. 
We construct a {\sc 2-SAT} formula $f_{G'}$ over the Boolean variables $x_i$, $i\in [r]$, with the following clauses:
\begin{itemize}
	\item For all $i \in [r]$, if $G'[N_{A,i}]$ contains an edge, then we add the clause $(x_i)$, and if $G'[N_{B,i}]$ contains an edge, then we add the clause $(\bar{x}_i)$.
	\item For every two distinct $i,j \in [r]$, if there is an edge between $N_{A,i}$ and $N_{A,j}$ in $G'$, then we add the clause $(x_i \vee x_j)$, and if there is an edge between $N_{B,i}$ and $N_{B,j}$ in $G'$, then we add the clause $(\bar{x}_i \vee \bar{x}_j)$.
\end{itemize}
\begin{lemma} \label{lem_2sat}
In the above setting, there is an independent cutset $S' \subseteq N_A \cup N_B$ separating $A$ and $B$ in $G'$ if and only if $f_{G'}$ is satisfiable.
\end{lemma}
\begin{proof}
Let $G'$, $f_{G'}$ and all corresponding sets be as in the statement.

For one direction, assume that there is an independent cutset $S' \subseteq N_A \cup N_B$ separating $A$ and $B$ in $G'$. 
We claim that either $K_i \cap S' = N_{A,i}$ or $K_i \cap S' = N_{B,i}$ is true for every $i \in [r]$.
Assume, for a contradiction, that the opposite is true for some $i \in [r]$. 
Let $N_{A,i}' = N_{A,i} \cap S'$ and $N_{B,i}' = N_{B,i} \cap S'$.
Then both sets are nonempty. Recall that $N_G(A) = N_A$, $N_G(B) = N_B$, $H[K_i]$ is connected and $S'$ is a cutset separating $A$ and $B$ in $G'$. 
Consider $d = \dist_H(N_{A,i}', N_{B,i}')$.
Since $H[K_i]$ is connected, $d < \infty$, since $H$ is bipartite, $d$ is odd, and since $S'$ is independent, $d \geq 2$.
Altogether, we have $d \geq 3$.
Since every vertex of $N_{A,i}'$ has a neighbor in $A$, and every vertex of $N_{B,i}'$ has a neighbor in $B$, the set $S'$ does not separate $A$ and $B$, a contradiction.
Using the claim, the assignment
\begin{align*}
	\alpha(x_i) = \begin{cases}
	\texttt{false}, &\text{if $K_i \cap S' = N_{A,i}$} \\
	\texttt{true}, &\text{if $K_i \cap S' = N_{B,i}$}
	\end{cases} \text{ for } i \in [r]
\end{align*}
is well-defined, and it satisfies $f_{G'}$.

For the opposite direction, let $\alpha$ be a satisfying assignment for $f_{G'}$. Then
\begin{align*}
S' = \left( \bigcup_{i \in [r]: \alpha(x_i)=\texttt{false}} N_{A,i} \right) \cup \left( \bigcup_{i \in [r]: \alpha(x_i)=\texttt{true}} N_{B,i} \right)
\end{align*}
is a subset of $N_A \cup N_B$, and it is easy to verify that it is an independent cutset of $G'$. This completes the proof. 
\QED \end{proof}

\begin{lemma} \label{lem_domindcs2}
Let $G$ be a connected graph, and let $X$ be a dominating set of $G$. 
If $G$ has an independent cutset $S^*$ such that the vertices of $X \setminus S^*$ are in at least two different components of $G - S^*$, then an independent cutset of $G$ can be found in $\calO^*(3^k)$ time, where $k = |X|$.
\end{lemma}
\begin{proof}
Let $G$, $X$ and $k$ be as in the statement.
If $G$ has an independent cutset $S^*$ such that the vertices of $X \setminus S^*$ are in at least two different components of $G - S^*$, then there is a partition $(A^*,B^*,X^*)$ of $X$ with the following properities.
We have $X^*=X\cap S^*$, the set $A$ is a maximal subset of vertices of $X$ contained in one component of $G - S^*$, and $B=X\setminus (X^*\cup A)$.

In $\calO^*(3^k)$ time one can enumerate all partitions of $X$ into three sets $A$, $B$ and $X'$ such that $A$ and $B$ are nonempty, there is no edge between $A$ and $B$, and $X'$ is an independent set of $G$. In other words, $X'$ is an independent cutset of $G[X]$ separating the nonempty sets $A$ and $B$.

For each enumerated partition $(A,B,X')$, let 
\begin{itemize}
    \item $N=N_{G}(A) \cap N_{G}(B)$,
    \item ${N_A=N_{G}(A)\setminus N}$,
    \item $N_B = N_{G}(B) \setminus N$, and
    \item let $H$ be the bipartite subgraph of $G$ induced by the edges between $N_A$ and $N_B$.
\end{itemize}

For the partition $(A,B,X')=(A^*,B^*,X^*)$ it holds that $(X' \cup N)\subseteq S^*$.
Thus, $I=X' \cup N$ must be an independent set of $G$ and $N_G(X' \cup N) \cap S^* = \emptyset$. 
Initialize $F=N_G(X' \cup N)$. 
($F$ is the set of vertices forbidden to be in $S^*$.) 
Clearly, $H$ has no edge with both endpoints in $F$; otherwise, we are not dealing with the partition $(A,B,X')=(A^*,B^*,X^*)$. 
Also, 
for any edge $uv \in E(H)$ such that $u\in F$ it holds that $v\in S^*$; otherwise $S^*$ does not separate $A$ and $B$. 
So, one can add $v$ into $I$. 
Note that $I$ is a subset of the vertices that must be in $S^*$. So, reversely, any edge of $H$ with one endpoint in $I$ must have the other endpoint forbidden to be in $S^*$, thus we can add it to $F$. This describes the rule to construct a bipartition $(I_H,F_H)$ of the vertices of $H$ that are connected to some vertex of $N_G(X' \cup N)$ in $H$. In particular, $I_H\subseteq I$, $F_H\subseteq F$, and $I$ must be an independent set of $G$.

Let $F' = F \setminus (N_A \cup N_B)$. Note that the vertices of $F'$ have neighbors in $X'$ but no neighbor in $A \cup B$, also, there may exist paths of $G$ from $A$ to $B$ passing through $F'$. For any path from $A$ to $B$ passing through $F'$ having only one vertex of $N_A$, say $a$, and only one vertex of $N_B$, say $b$, it holds that either $a$ or $b$ must be in $S^*$. The same holds with edges $ab \in E(G)$ such that $a\in N_A$ and $b\in N_B$. Also, if there is a path of $G$ from $A$ to $B$ passing through $F'$, then there is a path from $A$ to $B$ passing through $F'$ having only one vertex of $N_A$ and only one vertex of $N_B$.

At this point, let $G'$ be the graph obtained from $G$ by
\begin{itemize}
    \item contracting all components of $A$ and $B$ into a single vertex,
    \item inserting all possible edges between $N_G(K) \cap N_A$ and $N_G(K) \cap N_B$, where $K$ is a component of $F'$, and
    \item deleting all vertices of $I \cup F$.
\end{itemize}

Now, according to Lemma~\ref{lem_2sat}, we can use a {\sc 2-SAT} formula to decide in polynomial time (cf.~\cite{aspvall_et_al_1979}) which vertices of $N_A$ and $N_B$ should be in $S^*$. These vertices, together with the vertices previously fixed in $I$, form an independent cutset of $G$ separating $A$ and $B$. 

Therefore, by checking all partitions $(A,B,X')$ of $X$, we can in $\calO^*(3^k)$ time either find the required cutset, or conclude that $G$ does not admit such a cutset.
This completes the proof.
\QED \end{proof}

\version{The entire procedure described in Lemma~\ref{lem_domindcs2} is presented as Algorithm~\ref{alg_domindcs2} in Appendix~\ref{app_domindcs2}.}{}

\medskip

Now we are in a position to formulate Theorem~\ref{thm_domindcs}, which is a direct consequence of Lemmata~\ref{lem_domindcs1} and~\ref{lem_domindcs2}.

\begin{theorem} \label{thm_domindcs}
Let $G$ be the connected input graph, and let $X$ be a dominating set of $G$. 
Then \independentcutset{} with $X$ as an additional input can be solved in time $\calO^*(3^k)$, where $k = |X|$.
\end{theorem}

As a corollary of Theorem~\ref{thm_domindcs}, we obtain an \FPT-algorithm for \independentcutset{} parameterized by the independence number that has single-exponential dependence on the parameter. The fixed-parameter tractability of \independentcutset{} parameterized by the independence number also follows from the result of Marx, O'Sullivan and Razgon~\cite{MaOsRa}; however, their resulting dependence on the parameter is larger, and they only consider the $s$-$t$ version of the problem. 

\begin{corollary} \label{cor_ind}
\independentcutset{} can be solved in $\calO^*(3^k)$ time, where $k$ is the independence number of the input graph.
\end{corollary}
\begin{proof}
This follows from Theorem~\ref{thm_domindcs}, because any maximal independent set is a dominating set. 
\QED \end{proof}

\subsection{Distance to bipartite graphs} \label{sec_bipartite}
Let $G$ be a graph.
An \emph{odd cycle transversal} of $G$ is a set of vertices $X$ such that $G-X$ is bipartite.
The minimum cardinality of such a set can be seen as a distance measure of how far away a graph is from being bipartite.
Given $G$ and a nonnegative integer $k$, it is \NP{}-hard to decide whether $G$ admits an odd cycle transversal of size at most $k$~\cite{garey_johnson_1990}.
Nevertheless, Reed, Smith and Vetta~\cite{reed_smith_vetta_2004} proved that it is fixed-parameter tractable with respect to the solution size.
Lokshtanov, Narayanaswamy, Raman, Ramanujan and Saket~\cite{lokshtanov_et_al_2014} showed that it can be decided in time $\calO^*(2.3146^k)$ if $G$ admits an odd cycle transversal of size at most $k$.
If the answer is affirmative, an odd cycle transversal of size at most $k$ can be determined as a byproduct.
With this, we are able to formulate Corollary~\ref{cor_bip}.

\begin{corollary} \label{cor_bip}
\independentcutset{} can be solved in $\calO^*(3^k)$ time, where $k$ is the odd cycle transversal number of $G$.
\end{corollary}
\begin{proof}
If a vertex is not part of a triangle in $G$, then we return its neighborhood as an independent cutset of $G$.
Therefore we may assume every vertex is part of a triangle in $G$.
We compute an odd cycle transversal $X$ of size at most $k$ of $G$.
As stated, this is fixed-parameter tractable with respect to $k$ as a parameter, and can be done faster than $\calO^*(3^k)$ time.
Since every vertex is part of a triangle in $G$, the set $X$ is a dominating set of $G$.
Therefore the statement follows from Theorem~\ref{thm_domindcs}.
\QED \end{proof}

A \emph{triangle-hitting set} of $G$ is a set of vertices whose removal makes the resulting graph triangle-free. It is easy to see that the previous proof works even if $X$ is a minimum triangle-hitting set instead of a minimum odd cycle transversal. Since a triangle-hitting set of $G$ having size at most $k$ can be found in $\calO^*(3^k)$ time, the following corollary also holds.

\begin{corollary} \label{cor_triangle}
\independentcutset{} can be solved in time $\calO^*(3^k)$, where $k$ is the size of a minimum triangle-hitting set of $G$.
\end{corollary}

\subsection{Distance to chordal graphs} \label{sec_chordal}

We begin by proving a more general theorem. 

\version{}{\input{sections/preliminaries-refined-td}}

\begin{theorem} \label{thm_resind}
Let $G$ be the connected input graph, and let $\hat\calT = (T, \{(X_t, U_t):t \in V(T)\})$ be a nice $\ell$-refined tree decomposition of $G$ with residual independence number at most $k$.
\independentcutset{} with $\hat\calT$ as an additional input can be decided in $\calO^*(3^\ell (2n)^k)$ time.
\end{theorem}
\begin{proof}
Let $G$ and $\hat{\calT}$ be as in the statement, and let $r$ be the root of $T$.
Let $V_t$ denote the union of all sets $X_{t'}$ such that $t' \in V(T)$ is a descendant of $t$ or $t'=t$.
We call a partition $(S, A, B)$ of $X_t$ such that $S$ is independent and there is no edge crossing from $A$ to $B$ in $G$ a \emph{potential $t$-partition}.

We outline a dynamic programming algorithm that operates on the nice tree decomposition $\hat{\calT}$ in a bottom-up manner. 
At each node $t \in V(T)$, and for all potential $t$-partitions, we compute a Boolean variable $\cut[t, S, A, B]$ that is \textup{\texttt{true}} if and only if there is an independent cutset $S^* \supseteq S$ of $G$ such that $S^* \cap X_t = S$, and $A$ and $B$ are in distinct components of $G[V_t] - S^*$.
After computing these variables, we return $\cut[r, \emptyset, \emptyset, \emptyset]$, which gives the correct answer.
\version{For the sake of exposition, we explain how to compute these variables in Appendix~\ref{app_dp}.}{For the sake of exposition, we refer the reader on how to compute these variables to our arxiv article~\cite{rauch2023exact}.}


The key property of the algorithm is the following: Given that $|U_t| \leq \ell$, there are at most $3^\ell$ ways to partition $U_t$ into three sets. 
Since $\alpha(G[X_t \setminus U_t]) \leq k$, it holds that for any set $S \subseteq X_t \setminus U_t$, the number of components of $G[X_t \setminus (U_t\cup S)]$ is at most $k$.
Thus, there are at most $\calO(n^k)$ independent subsets $S$, and there are at most $2^k$ ways to partition the components of $G[X_t \setminus (U_t \cup S)]$. 
The computational effort to compute one variable $\cut[t, S, A, B]$ is polynomial.
Therefore the overall running time is $\calO^*(3^l(2n)^k)$. 
At this point, it is not hard to see how to perform such a dynamic programming using a nice tree decomposition within the claimed running time. 
This completes the proof of Theorem~\ref{thm_resind}. 
\QED \end{proof}

Let $G$ be a graph.
A \emph{chordal (vertex-)deletion} of $G$ is a set of vertices $X$ such that $G-X$ is chordal.
The problem of deciding whether a graph admits a chordal deletion of a fixed size $\ell$ is \NP{}-complete~\cite{garey_johnson_1990}.
Marx~\cite{marx_2010} showed that this problem is fixed-parameter tractable when parameterized by $\ell$.

\begin{corollary}
Let $G$ be the connected input graph, and let $\ell$ be the size of a chordal deletion of $G$.
If $\calO^*(f(\ell)$ is the running time of an algorithm computing a chordal deletion of size $\ell$ of a graph,
then \independentcutset{} can be solved in $\calO^*(f(\ell) + 3^\ell)$ time.
\end{corollary}
\begin{proof}
First, we compute a chordal deletion $X$ of size $\ell$ of $G$.
Then, since $G-X$ is chordal, we compute a clique tree $\calT$ of $G-X$, which is possible in polynomial time~\cite{blair_peyton_1993}. 
To obtain a tree decomposition of $G$, we add the vertices of $X$ to every bag of $\calT$. 
We transform this tree decomposition into a nice tree decomposition $\hat\calT = (T, \{X_t: t \in V(T) \})$ in polynomial time
\version{, which is possible due to Lemma~\ref{lem_nice}.}{.}
Note that this can be done in a way such that for every node $t \in V(T)$, the set $X_t \setminus X$ still induces a clique in $G$.
By letting $U_t = X \cap X_t$ for every node $t \in V(T)$, we augment $\hat\calT$ to a nice $\ell$-refined tree decomposition with residual independence number 1.
Now the statement follows from Theorem~\ref{thm_resind}. 
\QED \end{proof}

\subsection{Distance to $P_5$-free graphs} \label{sec_p5}
A \emph{$P_5$-hitting set} of $G$ is a set of vertices $X$ such that $G-X$ is $P_5$-free.
In this subsection, we consider \independentcutset{} parameterized by the size of a $P_5$-hitting set.
Bacsó and Tuza~\cite{bacso_tuza_1990} showed that any connected $P_5$-free graph has a dominating clique or a dominating $P_3$.
Camby and Schaudt~\cite{camby_schaudt_2016} generalized this result and showed that such a dominating set can be computed in polynomial time.
We use the statement of the following corollary to prove a stronger statement in Theorem~\ref{thm_p5}.

\begin{corollary} \label{cor_p5}
\independentcutset{} can be solved in polynomial time for connected $P_5$-free graphs.
\end{corollary}
\begin{proof}
Let $G$ be a connected $P_5$-free graph. 
We compute a dominating set $X$ as in~\cite{camby_schaudt_2016} in polynomial time, that is, the set $X$ induces either a clique or a $P_3$ in $G$.
We invoke the algorithm of Theorem~\ref{thm_domindcs} with $G$ and $X$ as input.
Since $X$ induces a clique or a $P_3$ in $G$, there are at most $\calO(n)$ partitions $(A, X')$ or $(A, B, X')$ of $X$ such that there is no edge between $A$ and $B$ in $G$, and $X'$ is independent in $G$. 
Given that the relevant partitions of $X$ can be enumerated in polynomial time, the remaining steps of the algorithm of Theorem~\ref{thm_domindcs} can solve \independentcutset{} in polynomial time for $P_5$-free graphs. 
\QED \end{proof}

For Theorem~\ref{thm_p5} we need the following fact.
\begin{proposition} \label{prop_p5}
A $P_5$-hitting set of $G$ with size $k$ (if any) can be found in \FPT-time with respect to $k$.
\end{proposition}
\begin{proof}
A simple bounded search tree algorithm has $\calO^*(5^k)$ running time. 
\QED \end{proof}

We will now prove the main theorem of this subsection.
\begin{theorem} \label{thm_p5}
If a $P_5$-hitting set with at most $k$ vertices of a connected graph $G$ can be computed in $\calO^*(f(k))$ time, then \independentcutset{} can be solved in $\calO^*(f(k) + 3^k)$ time.
\end{theorem}
\begin{proof}
Let $G$ and $k$ be as in the statement. 
First, we compute a $P_5$-hitting set $X$ of size at most $k$ of $G$.
Let $K_1, \dots, K_\ell$ be the components of $G - X$.
Then we compute a dominating set $X_i$ of $G[K_i]$ as in \cite{camby_schaudt_2016}, that is, the set $X_i$ induces either a clique or a $P_3$ in $G[K_i]$.
Note that $X_i \cup X$ dominates $G[K_i \cup X]$ for every $i \in [\ell]$.
Assume that $G$ admits an independent cutset $S^*$.
There are two cases: Either $X \setminus S^*$ is contained in at most one component of $G - S^*$, or not.
We explain for both cases how to compute an independent cutset of $G$ in time $\calO^*(3^k)$.

\emph{Case 1:} $X \setminus S^*$ is contained in at most one component of $G - S^*$. 
This implies that there exists some $i \in [\ell]$ such that $S^* \cap (K_i \cup X)$ is an independent cutset of $G[K_i \cup X]$. 
We iterate over all disjoint partitions $(A^*, X^*)$ of $X$ such that $X^*$ is independent in $G$. 
In one iteration, we guess $S^* \cap X$ as $X^*$.
Now we invoke a modified version of the algorithm of Theorem~\ref{thm_domindcs} with $G[K_i \cup X]$ and $X_i \cup X$ as input.
This version only considers partitions $(A, X')$ or $(A, B, X')$ of $X$ with $A^* \subseteq A$ and $X^* \subseteq X'$.
As in Corollary~\ref{cor_p5}, this takes only polynomial time since the number of relevant partitions is polynomial.
By Theorem~\ref{thm_domindcs}, the modified algorithm returns an independent cutset of $G[K_i \cup X]$, which is also an independent cutset of $G$.
The described procedure can be implemented to run $\calO^*(2^k)$ time.

\emph{Case 2:} $X \setminus S^*$ is contained in at least two components of $G - S^*$. 
We iterate over all disjoint partitions $(A^*, B^*, X^*)$ of $X$ such that $X^*$ is an independent cutset of $G[X]$ that separates $A^* \neq \emptyset$ and $B^* \neq \emptyset$.
In one iteration, we guess $A^*$ and $B^*$ such that they are in different components of $G-S^*$, and we guess $S^* \cap X$ as $X^*$.
This time, we invoke a modified version of the algorithm of Lemma~\ref{lem_domindcs2} with $G[K_i \cup X]$ and $X_i \cup X$ as input.
This version only considers partitions $(A, B, X')$ with $A^* \subseteq A$, $B^* \subseteq B$ and $X^* \subseteq X'$.
As in Corollary~\ref{cor_p5}, this takes only polynomial time. 
The modified algorithm returns, for every $i \in [\ell]$, an independent cutset $S_i$ of $G[X_i \cup X]$ with the following properties: the set $S_i$ separates $A^*$ and $B^*$ in $G[X_i \cup X]$, and $X^* \subseteq S_i$. Now, $\bigcup_{i \in \ell} S_i$ is an independent cutset of $G$. The described procedure for this can be implemented to run in $\calO^*(3^k)$ time.

The overall running time is $\calO^*(f(k) + 3^k)$, which completes the proof.
\QED \end{proof}

\subsection{Generalizing distance to $P_5$-free graphs} \label{sec_gen_p5}
In this section, we generalize Corollary~\ref{cor_p5} and Theorem~\ref{thm_p5}.
Let $G$ be a graph, and let $\mathcal{G}$ be a class of graphs.
We say that a set of vertices $X$ is a \emph{$\alpha_k$-dominating set} of $G$ if $X$ is a dominating set of $G$ and $\alpha(G[X]) \leq k$.
If $G$ admits a $\alpha_k$-dominating set, we say that $G$ is \emph{$\alpha_k$-dominated}. In addition, if such a set can be computed in polynomial time, we say that $G$ is \emph{efficiently $\alpha_k$-dominated}.
The \emph{(efficient) $\alpha$-domination number} of $G$ is the minimum $k$ such that $G$ is (efficient) $\alpha_k$-dominated.
We say that a graph class $\mathcal{G}$ is (efficiently) $\alpha_k$-dominated if every graph $G \in \mathcal{G}$ is (efficiently) $\alpha_k$-dominated.
For example, Bacsó and Tuza~\cite{bacso_tuza_1990} as well as Cozzens and Kelleher~\cite{cozzens_kelleher_1990} independently proved that $\{P_5, C_5\}$-free graphs are $\alpha_1$-dominated, that is, they contain a dominating clique.
Another example is the class of $P_5$-free graphs, which is efficiently $\alpha_2$-dominated. 

We start by generalizing the polynomial time result on $P_5$-free graphs. 
The proof is similar to the one of Corollary~\ref{cor_p5}.

\begin{lemma}\label{lem_alpha-dominated}
\independentcutset{} on $\mathcal{G}$ can be solved in polynomial time, whenever $\mathcal{G}$ is an efficiently $\alpha_c$-dominated graph class for some constant $c$. 
\end{lemma}

In~\cite{penrice_1995}, Penrice presented some families of $\alpha_k$-dominated graph classes. 
For example, Penrice shows that connected $\{P_6, H_{t+1}\}$-free graphs are $\alpha_t$-dominated, where $H_{t+1}$ denotes the graph obtained by subdividing each edge of a $K_{1,{t+1}}$.
We are currently unaware if they are efficiently $\alpha_t$-dominated, too, since Penrice uses a minimality argument of Bacs\'o and Tuza~\cite{bacso_tuza_1990}.
Penrice also shows that $tK_2$-free graphs without isolated vertices are $\alpha_{2t-2}$-dominated. 
From a thorough reading in~\cite{penrice_1995}, it can be seen that $tK_2$-free graphs without isolated vertices are efficiently $\alpha_{2t-2}$-dominated. Thus, from Lemma~\ref{lem_alpha-dominated} and the results in~\cite{penrice_1995} the following holds.  

\begin{corollary} \label{cor_p6_fork}
\independentcutset{} can be solved in polynomial time for ${tK_2}$-free graphs for any constant $t \geq 1$.
\end{corollary}

It has been remarked that Prisner proved in \cite{prisner1995graphs} that, for any integer $t$, any $tK_2$-free graph has polynomially many maximal independent sets.
We were not able to verify this, because his article was inaccessible to us.
With this, the statement of Corollary~\ref{cor_p6_fork} already follows from Corollary~\ref{cor_exact_general}.

Theorem~\ref{thm_gen_p5} is a generalization of Theorem~\ref{thm_p5} and Lemma~\ref{lem_alpha-dominated}, and its proof is similar to the proof of Theorem~\ref{thm_p5}.

\begin{theorem} \label{thm_gen_p5}
Let $c$ be a fixed constant and $G$ be the input graph.
If a set $X$ of $k$ vertices such that $G-X$ is efficiently $\alpha_c$-dominated can be found in $\calO^*(f(k))$ time, then \independentcutset{} can be solved in $\calO^*(f(k) + 3^k)$ time.
\end{theorem}

%% file: sections/algorithm-domindcs2.tex
\section{The Algorithm of Lemma~\ref{lem_domindcs2}} \label{app_domindcs2}

\begin{algorithm}[ht]
\SetKw{All}{all}
\SetKw{Continue}{continue}
\KwIn{A connected graph $G$, and a dominating set $X$ of $G$.}
\SetKwInOut{Parameter}{Parameter}
\Parameter{$k = |X|$.}
\KwOut{Either an independent cutset $S$ of $G$, or the statement ``no such independent cutset''.}
\BlankLine
\For({(i)}){\All disjoint partitions $(A, B, X')$ of $X$ such that $X'$ is an independent cutset of $G[X]$ separating the nonempty sets $A$ and $B$}{
	let $N \gets N_{G}(A) \cap N_{G}(B)$, $N_A \gets N_{G}(A) \setminus N$ and $N_B \gets N_{G}(B) \setminus N$\;
	let $H$ be the bipartite subgraph of $G$ induced by $N_A$ and $N_B$\;
	initialize $I \gets X' \cup N$ and $F \gets N_G(I) \setminus X$\;
	\While({(ii)}){$\exists uv \in E(H): u \in F, v \notin F \text{ and } v \notin I$}{
		update $I \gets I \cup \{v\}$ and $F \gets F \cup N_H(v)$\;
	}
	\If({(iii)}){$I$ is not an independent set of $G$ or $\exists uv \in E(H): u, v \in F$}{\Continue\;}
	let $F' \gets F \setminus (N_A \cup N_B)$\;
	let $G'$ be the graph obtained from $G$ by 1. contracting all components of $G[A]$ respective $G[B]$ to a single vertex, 2. inserting all possible edges between $N_G(K) \cap N_A$ and $N_G(K) \cap N_B$ for all components $K$ of $G[F']$, and 3. deleting all vertices of $I \cup F$\;
	construct a $2$-SAT formula $f_{G'}$ according to Lemma~\ref{lem_2sat} for $G'$\;
	\If({(iv)}){$f_{G'}$ is satisfiable}{
		construct an independent cutset $S'$ of $G'$ as in the proof of Lemma~\ref{lem_2sat}\;
		\Return $I \cup S'$\;
	}
	
}
\Return ``no such independent cutset''\;
\BlankLine
\caption{The algorithm of Lemma~\ref{lem_domindcs2}.} \label{alg_domindcs2}
\end{algorithm}

%% file: sections/dynamic-programming.tex
\section{Dynamic Programming for Theorem~\ref{thm_resind}} \label{app_dp}
Let $G$ and $\hat{\calT}$ be as in the statement of Theorem~\ref{thm_resind}, and let $r$ be the root of $T$.
Let $V_t$ denote the union of all sets $X_{t'}$ such that $t' \in V(T)$ is a descendant of $t$ or $t'=t$.
We call a disjoint partition $(S, A, B)$ of $X_t$ such that $S$ is an independent set of $G$ a \emph{potential $t$-partition}.

The Boolean variables $\cut[t, S, A, B]$ are computed as follows:
\begin{itemize}
\item If $t$ a leaf, then set $\cut[t, \emptyset, \emptyset, \emptyset] \gets \texttt{false}$.
\item If $t$ is an introduce node with child $t'$, let $v$ be the vertex introduced by $t$. 
For all potential $t'$-cuts $(S', A', B')$, set $\cut[t, S', A' \cup \{v\}, B']$ to \texttt{true}
\begin{itemize}
\item if $\cut[t', S', A', B'] = \texttt{true}$ and $N_{G}(v) \cap B' = \emptyset$,
\item if $\cut[t', S', A', B'] = \texttt{false}$, $A'=\emptyset$, $B'\neq\emptyset$ and $N_{G}(v) \cap B' = \emptyset$, or
\item if $\cut[t', S', A', B'] = \texttt{false}$, $A'=B'=\emptyset$ and some descendant of $t$ is a forget node.
\end{itemize}
Otherwise set it to \texttt{false}.
Set $\cut[t, S', A', B' \cup \{v\}]$ analogously. 
Furthermore, if $S' \cup \{v\}$ is an independent set of $G$, then set $\cut[t, S' \cup \{v\}, A', B']$ to $\cut[t', S', A', B']$.
\item If $t$ is a forget node with child $t'$, let $w$ be the vertex forgotten by $t$.
For all potential $t$-partitions $(S, A, B)$, set $\cut[t, S, A, B]$ to
\begin{itemize}
\item $\cut[t', S, A \cup \{w\}, B] \vee \cut[t', S, A, B \cup \{w\}] \vee \cut[t', S \cup \{w\}, A, B]$, if $S \cup \{w\}$ is an independent set of $G$, and
\item $\cut[t', S, A \cup \{w\}, B] \vee \cut[t', S, A, B \cup \{w\}]$, otherwise.
\end{itemize}
\item If $t$ is a join node with children $t'$ and $t''$, then for all potential $t$-partitions $(S, A, B)$, set $\cut[t, A, B, S]$ to
\begin{itemize}
\item $\cut[t', A, B, S] \wedge \cut[t'', A, B, S]$, if $A \neq \emptyset$ and $B \neq \emptyset$,
\item $\cut[t', A, B, S] \vee \cut[t'', A, B, 
S]$, if $A \neq \emptyset$ or $B \neq \emptyset$,
\item \texttt{true}, if $A = B = \emptyset$ and there are some descendants of both $t'$ and $t''$ that are forget nodes,
\item \texttt{false}, otherwise.
\end{itemize}
\end{itemize}

In Theorem~\ref{thm_resind} we made the following claim, which we will now prove.
\begin{claim}
The Boolean variable $\cut[t, S, A, B]$ is \textup{\texttt{true}} if and only if there is an independent cutset $S^* \supseteq S$ of $G$ such that $S^* \cap X_t = S$, and $A$ and $B$ are in distinct components of $G[V_t] - S^*$.
\end{claim}
\begin{proof}
We prove the statement by induction on $T$. 
Let $t$ be a node of $T$. 
The base case is when $t$ is a leaf, which is clearly true. 
In the inductive case, the node $t$ is either an introduce, forget or join node. 
Let $(S, A, B)$ be a potential $t$-partition.
We call a set $S^*$ as in the statement a \emph{$(t, S, A, B)$-cutset}.

\emph{Case 1}: $t$ is an introduce node. 
Let $t'$ be the child of $t$, and let $v$ be the vertex introduced by $t$.
The set $X_t \cap X_{t'} = X_{t'}$ separates $V_{t'}$ and $V(G) \setminus V_{t'}$ in $G$.
In particular, the vertex $v$ has no neighbors in $V_{t'} \setminus X_{t'}$.
We will use this in the following without referring to it.

Assume that $\cut[t, S, A, B] = \texttt{true}$, and consider the case $v \in A$. 
If $\cut[t', S, A \setminus \{v\}, B] = \texttt{true}$, then $N_G(v) \cap B = \emptyset$. 
By the inductive hypothesis, there exists a $(t', S, A \setminus \{v\}, B)$-cutset. This is also a $(t, S, A, B)$-cutset.
If $\cut[t', S, A\setminus\{v\}, B] = \texttt{false}$, then $S$ is a $(t, S, A, B)$-cutset in both cases.
The case $v \in B$ is analogous. 
Consider the case $v \in S$. By the inductive hypothesis, there is a $(t', S \setminus \{v\}, A, B)$-cutset $S^*$. 
Since $S^* \cap X_{t'} = S \setminus \{v\}$, the set $S^* \cup \{v\}$ is an independent set of $G$.
So $S^* \cup \{v\}$ is a $(t, S, A, B)$-cutset.

For the opposite direction, assume that $S^*$ is a $(t, S, A, B)$-cutset. Consider the case $v \in A$. 
Then we have $N_G(v) \cap B = \emptyset$. If $\cut[t', S, A \setminus \{v\}, B]$ is \texttt{true}, then so is $\cut[t, S, A, B]$. 
Otherwise, by the induction hypothesis, the set $S^*$ is not a $(t', S, A \setminus \{v\}, B)$-cutset. 
This implies $A = \{v\}$.
Since $S^*$ is a $(t, S, A, B)$-cutset, we must have $N_G(v) \cap B = \emptyset$. 
Furthermore, since $V_{t'} \setminus (S^* \cup A)$ is nonempty, we have $B \neq \emptyset$ or there is some already forgotten vertex in $V_{t'} \setminus (S^* \cup A)$. 
Therefore, the algorithm sets $\cut[t, S, A, B]$ to \texttt{true}. 
The case $v \in B$ is similar. 
Consider the case $v \in S$.
Here $S^* \setminus \{v\}$ is a $(t', S \setminus \{v\}, A, B)$-cutset.
Therefore $\cut[t', S \setminus \{v\}, A, B] = \texttt{true}$ by the induction hypothesis, and the algorithm sets $\cut[t, S, A, B]$ to \texttt{true} too.

\emph{Case 2}: $t$ is a forget node. 
Let $t'$ be the child of $t$, and let $w$ be the vertex forgotten by $t$.
Here $X_t \cap X_{t'} = X_t$ separates $V_{t'}$ and $V(G) \setminus V_{t'}$ in $G$.

Assume that $\cut[t, S, A, B] = \texttt{true}$. 
Then $\cut[t', S, A \cup \{w\}, B]$, $\cut[t', S, A, B\cup \{w\}]$ or, if $S \cup \{w\}$ is an independent set of $G$, $\cut[t', S \cup \{w\}, A, B]$ is \texttt{true}. 
By the induction hypothesis, there is a $p$-cutset, where $p$ is the respective 4-tuple. 
This is $p$-cutset is also a $(t, S, A, B)$-cutset.

For the other direction, assume that $S^*$ is a $(t, S, A, B)$-cutset. 
If $w \in S^*$, then $S^*$ is a $(t', S \cup \{w\}, A, B)$-cutset. 
Otherwise, the vertex $w$ cannot have neighbors both in $A$ and in $B$. 
Therefore $S^*$ is a $(t', S, A \cup \{w\}, B)$-cutset or a $(t', S, A, B \cup \{w\})$-cutset.
By induction, one of the variables $\cut[t', S, A \cup \{w\}, B]$, $\cut[t', S, A, B \cup \{w\}]$ or, if $S \cup \{w\}$ is an independent set of $G$, $\cut[t', S \cup \{w\}, A, B]$ is \texttt{true}, and the algorithm sets $\cut[t, S, A, B]$ to \texttt{true}.

\emph{Case 3}: $t$ is a join node with children $t'$ and $t''$.
Here $X_t = X_{t'} = X_{t''}$ separates $V_{t'}$, $V_{t''}$ and $V(G) \setminus V_t$ in $G$.

Assume that $\cut[t, S, A, B] = \texttt{true}$.
Consider the case that $A \neq \emptyset$ or $B \neq \emptyset$, and without loss of generality let $A \neq \emptyset$.
If $B \neq \emptyset$, then $\cut[t', S, A, B] = \cut[t'', S, A, B] = \texttt{true}$. By induction there is a $(t', S, A, B)$-cutset $S_1^*$, and a $(t'', S, A, B)$-cutset $S_2^*$.
The union $S_1^* \cup S_2^*$ is a $(t, S, A, B)$-cutset. 
If $B = \emptyset$, then without loss of generality let $\cut[t', S, A, B] = \texttt{true}$. 
By the induction hypothesis, there is a $(t', S, A, B)$-cutset. 
This is also a $(t, S, A, B)$-cutset.
Now consider the case that $A = B = \emptyset$. By the definition of $\cut[t, S, A, B]$, the set $S = X_t$ an independent set of $G$ and it is a $(t, S, A, B)$-cutset.

For the other direction, let $S^*$ be a $(t, S, A, B)$-cutset.
If $A \neq \emptyset$ and $B \neq \emptyset$, then $S^* \cap V_{t'}$ is a $(t', S, A, B)$-cutset, and $S^* \cap V_{t''}$ is a $(t'', S, A, B)$-cutset.
By induction, $\cut[t', S, A, B] = \cut[t'', S, A, B] = \texttt{true}$.
If exactly one of the sets $A$ or $B$ is empty (but not both), then one component of $G - S^*$ must be in $V_{t'} \setminus X_t$ or in $V_{t''} \setminus X_t$.
Without loss of generality let the former be the case.
This implies that $S^* \cap V_{t'}$ is a $(t', S, A, B)$-cutset.
By induction, $\cut[t', S, A, B] = \texttt{true}$.
If $A = B = \emptyset$ and $\cut[t', S, A, B] = \cut[t'', S, A, B] = \texttt{false}$, then $S^*$ is not a $(t', S, A, B)$-cutset nor a $(t'', S, A, B)$-cutset by induction.
Since it is a $(t, S, A, B)$-cutset, one component of $G - S^*$ must be in $V_{t'} \setminus X_t$, and another component of $G - S^*$ must be in $V_{t''} \setminus X_t$.
In particular, both $t'$ and $t''$ must have a forget node as a descendant.
Altogether, the algorithm sets $\cut[t, S, A, B]$ to \texttt{true}.

The proof of the claim is complete.
\QED \end{proof}